\documentclass{article}
\usepackage{amssymb}

\usepackage{amsmath}
\usepackage[utf8]{inputenc}
\usepackage{epsfig}

\usepackage{subcaption}

\usepackage{graphicx}
\usepackage{mathrsfs}
\usepackage{bbm}
\usepackage{float}
\usepackage{pifont}
\usepackage{color}

\usepackage{subcaption}

\usepackage{ctable,subfig}

\usepackage{multirow}
\usepackage{makecell}
\usepackage{mathrsfs}
\usepackage{bbm}
\usepackage{setspace}
\usepackage{float}
\usepackage[shortlabels]{enumitem}
\usepackage{pifont}

\usepackage{color}
\usepackage[normalem]{ulem}
\newtheorem{lemma}{Lemma}

\newtheorem{definition}{Definition}

\newtheorem{example}{Example}

\RequirePackage{lineno}

\begin{document}

{\bf Maintenance policy for a system with a weighted linear combination of degradation processes}

Shaomin Wu, Kent Business School, University of Kent (Corresponding author. Email: s.m.wu@kent.ac.uk) \\
Inma T. Castro, Department of Mathematics, University of Extremadura, Spain

\begin{abstract}
This paper develops maintenance policies for a system under condition monitoring. We assume that a number of defects may develop and the degradation process of each defect follows a gamma process, respectively. The system is inspected periodically and maintenance actions are performed on the defects present in the system. The effectiveness of the maintenance is assumed imperfect and it is modelled using a geometric process. By performing these maintenance actions, different costs are incurred depending on the type and the degradation levels of the defects. Furthermore, once a linear combination of the degradation processes exceeds a pre-specified threshold, the system needs a special maintenance and an extra cost is imposed. The system is renewed after several preventive maintenance activities have been performed. The main concern of this paper is to optimise the time between renewals and the number of renewals. Numerical examples are given to illustrate the results derived in the paper.

\end{abstract}

\section{Introduction}
Condition-based maintenance has been extensively studied in the reliability literature due to the emergence of advanced condition monitoring and data collection techniques. Many papers have been published to either model the degradation processes of assets \cite{Si1}, \cite{Ye}, \cite{Deng},\cite{Zhao1} or to optimise maintenance policies \cite{Caballe2},\cite{Liu},\cite{Zhao2}. To obtain a comprehensive view of the development in condition-based maintenance, the reader is referred to review papers, see \cite{Jardine},\cite{Si1},\cite{Alaswad}, for example.

A number of degradation processes have been considered in condition-based maintenance policies. Many authors investigate different maintenance policies, considering only one degradation process such as the gamma process \cite{Caballe2}, the Wiener process \cite{Sun}, the inverse Gaussian process \cite{Chen} and the Ornstein-Uhlenbeck process \cite{Deng}. Some consider condition-based maintenance policies for assets suffering a number of degradation processes. For example, \cite{Caballe1} proposes a condition-based maintenance strategy for a system subject to two 
dependent causes of failure, degradation and sudden shocks: The internal degradation is 
reflected by the presence of multiple degradation processes in the system, and degradation processes start at random times following a Non-homogeneous Poisson process and their growths are modelled by using a gamma process. \cite{Huynh} consider maintenance policies monitored by a process of the average of a number of degradation processes. 

In this paper, we consider a system on which many different types of defects develops over the time. 
If a linear combination of the degradation processes exceeds a pre-specified threshold in an inspection time, maintenance is carried out. There are many real-world examples behaving like that in the real world. For example, in the civil engineering, several different types of defects, such as fatigue cracking and pavement deformation, may develop simultaneously on a pavement network. The mechanism of these defects may be different: fatigue cracking is caused by the failure of the surface layer or base due to repeated traffic loading (fatigue), and pavement deformation is the result of weakness in one or more layers of the pavement that has experienced movement after construction \cite{Adlinge}. As such, the deteriorating processes of these defects are different in the sense that the parameters in the degradation processes may differ. Furthermore, both the approaches to repairing these defects and the cost of repairing them differ from defect to defect. The time to repair such a system may be the time when a linear combination of those defects exceeds a pre-specified threshold. In the civil engineering literature, for example, \cite{Shah} propose a linear combination of defects of pavement condition indexes and suggest that a pavement needs maintenance once its combined condition index exceeds a pre-specified threshold. It should also be noted that such deterioration might cause partial loss of system functionality. As such, there is no need to overhaul or renew the entire system unless its combined index exceeds a threshold that is large enough. 

Inspired from the above real world example, this paper develops maintenance policies for a system with many degradation processes. The system is inspected periodically. Following an inspection, an imperfect repair is performed. These imperfect repairs are modelled using geometric process. Costs of repairing different defects are different and, if the linear combination of the magnitudes of a set of defects exceeds a pre-specified threshold, an additional cost is incurred. A replacement is carried out once the number of inspections exceeds an optimum value.

The remainder of the paper is structured as follows. Section 2 introduces the notations and assumptions that will be used in the paper. Section 3 derives distributions of hitting time and considers random effects. Section 4 derives maintenance policies and proposes methods of optimisation. Section 5 illustrates the maintenance policies with numerical examples. {\color{black} Section 6 offers discussion on some of the assumptions in this paper.}  Section 7 concludes the paper.


	


\section{Assumptions}
\label{Assumptions}
This paper makes the following assumptions.
\begin{enumerate}[label=A\arabic*).]
	\item  Defects of $n$ types  develop
	 through $n$ degradation processes on a system, respectively. 
	\item The system is inspected every $T$ time units ($T>0$). 
	\item {\color{black} The system is new at time $t=0$.} 	
	\item Two types of maintenance are taken: an imperfect maintenance and a complete replacement of the system. The imperfect maintenance restores the system to a state between a good-as-new state (which is resulted from a replacement) and a bad-as-old state (which is resulted from a minimal repair) and is modelled using a geometric process. The replacement completely renews the system.
	\item On performing these maintenance actions, a sequence of costs is incurred. Repairing the $k$-th ($k=1, 2, \ldots, n$) defect incurs two types of cost: a fixed cost, and a variable cost that depends on the degradation level of the $k$-th defect. Furthermore, if the linear combination of the magnitudes of a set of defects exceeds a pre-specified threshold, an additional cost is incurred.
	\item Imperfect maintenance actions are performed every $T$ time units and preventive replacement is performed at the time of the $N$-th inspection.
	\item Maintenance time is so short that it can be neglected.
\end{enumerate}

\section{Model development}
\cite{Noortwijk} optimise inspection decisions for scour holes, on the basis of the uncertainties in the process of occurrence of scour holes and, given that a scour hole has occurred, of the process of current-induced scour erosion. The stochastic processes of scour-hole initiation and scour-hole development was regarded as a Poisson process and a gamma process, respectively. \cite{Lawless} construct a tractable gamma-process model incorporating a random effect and fit the model to some data on crack growth. In the following, we make similar assumptions: The stochastic processes of defect initiation and defect development was regarded as a Poisson process and a gamma process, respectively.

\subsection{Modelling the occurrences of the defects}
{\color{black} Denote the successive times between occurrences of the defects by the infinite sequence of non-negative real-valued random quantities $T_1,T_2,...$. }
Assume the defect initiation follows a homogeneous Poisson process. Similar to the assumptions made in \cite{Noortwijk}, we assume the defect inter-occurrence times are {\it exchangeable} and they exhibit the {\it memorylessness} property. That is, the order in which the defects occur is irrelevant and the probability distribution of the remaining time until the occurrence of the first defect does not depend on the fact
that a defect has not yet occurred since the last replacement.
{\color{black}
 According to \cite{Noortwijk}, the joint probability density function of $T_1,T_2,....,T_n$ is given by
\begin{equation}
p_{T_1,T_2,...,T_n}(t_1,...,t_n )=\int_0^{\infty} \prod_{k=1}^n \frac{1}{\lambda} \exp \left(-\frac{t_k}{\lambda} \right)  p_{\Lambda}(\lambda) d \lambda,
\end{equation}
where $(t_1,t_2,....,t_n) \in \mathbb{R}^n_+ $, $p_{\Lambda}(\lambda)=\frac{1}{\Gamma(\nu)}\mu^{\nu} \lambda^{-(\nu+1)}e^{-\nu/\lambda}1_{\{\lambda>0\}}$, where $\mu$ and $\nu$ are parameters that can be estimated from given observations, $1_{\{\lambda>0\}}=1$ if $\lambda>0$ and $1_{\{\lambda>0\}}=0$ otherwise.} With the constraint $T_1,T_2,....,T_n <T$, we assume that the $n$ defects  occur during time interval $(0,T)$. For those defects occurring within other time intervals $(kT,(k+1)T)$ (for $k$=1,2,...,), a similar joint probability density function can be derived.

\subsection{Degradation processes}
\label{Linear_Combination}
We consider the situation where $n$ types of defects may develop and their degradation processes $\{X_k(t), t \ge 0\}$ for $k=1,2,...,n$, respectively. That is, $X_k(t)$ is the deterioration level of the $k$th degradation process at time $t$ and \{$X_k(t)$, $k$=1,...,$n$\} are independent. 

Assume that $X_k(t)$ has the following properties:
\begin{enumerate}[a)]
\item $X_k(0)=0$,
\item the increments $\Delta X_k(t)=X_k(t+\Delta t) - X_k(t)$ are independent of $t$,
\item $\Delta X_k(t)$ follows a gamma distribution {\color{black} Gamma$(\alpha_k(t+\Delta t) -\alpha_k(t),\beta_k)$ with shape parameter $\alpha_k(t+\Delta t) -\alpha_k(t)$ and scale parameter $\beta_k$, where $\alpha_k(t)$ is a given monotone increasing function in $t$ and $\alpha_k(0)=0$}. 
\end{enumerate}
{\color{black} $X_k(t)$ has probability distribution Gamma$(\alpha_k(t),\beta_k)$ with mean $\beta_k \alpha_k(t)$ and variance $\beta_k^2 \alpha_k(t)$,} and its probability density function is given by
\begin{equation} \label{gammapdf} 
	f(x;\alpha_k(t), \beta_k)=\frac{\beta_k^{-\alpha_k(t)}}{\Gamma(\alpha_k (t))}x^{\alpha_k(t)-1}e^{-x/\beta_k }1_{\{x>0\}},
\end{equation}
where $\Gamma(.)$ is the gamma function. 

Suppose that the system needs maintenance as long as a linear combination of the magnitudes of the $n$ defects exceeds a pre-specified threshold. In reality, for example, a section of a pavement network may have more than $n$ defects, some of which may be of the same type. The pavement needs maintenance as long as the linear combination exceeds a pre-specified threshold.

We consider that the {\it overall degradation} of the system is represented by 
 \begin{equation}    \label{Y}
Y(t)=\sum_{k=1}^n b_k X_k(t), \quad t \geq 0, \quad b_k \geq 0, 
 \end{equation}  
{\color{black} where $b_k$ (with $b_k >0$) is the weight of defect $k$. Denote $Y_k(t)=b_k X_k(t)$. Then $Y(t)=\sum_{k=1}^n Y_k(t)$ and $Y_k(t)$ has pdf $f(x;\alpha_k(t), b_k\beta_k)$.}

Then the expected value and the variance of $Y(t)$ are given by
 \begin{equation}   
\mathbb{E}(Y(t))=\sum_{k=1}^n b_k \beta_k \alpha_k(t),
\label{Y_Mean}
 \end{equation} and
 \begin{equation}   
\textrm{var}(Y(t))=\sum_{k=1}^n b^2_k \beta_k^2 \alpha_k(t), 
\label{Y_Variance}
 \end{equation}  
respectively. 

Furthermore, the overall degradation process $\{Y(t), t \geq 0\}$, given by Eq. \eqref{Y} is a stochastic process with the following properties.
\begin{enumerate}[a)]
\item $Y(0)=\sum_{k=1}^n b_k X_k(0)=0$,
\item If the increment $\Delta X_k(t)=X_k(t+\Delta t) - X_k(t)$ is independent of $t$, then $\Delta Y(t)=\sum_{k=1}^n b_k \Delta X_k(t)$ is independent of $t$ as well,
\end{enumerate} 
According to \cite{Moschopoulos}, the density function of $Y(t)$ can be expressed by
\begin{equation}   
	g_{Y(t)}(y)=D(t) \sum_{k=0}^{\infty}\frac{\zeta_{k}(t) \beta_0^{-\rho(t)-k}}{\Gamma(\rho(t) + k)} y^{\rho(t) +k -1} e^{-y/\beta_0}, \quad y > 0, 
	\label{Y_distribution} 
\end{equation}  
where $\beta_0=\min_{1 \le k \le n} b_k \beta_k$. $D(t)$ and $\rho(t)$ are given by 
 \begin{equation} \label{beta0D0}
D(t)=\prod_{k=1}^n ( \frac{\beta_0}{b_k \beta_k })^{\alpha_k(t)},
\end{equation}   
and
 \begin{equation} \label{beta0D}
\rho(t)=\sum_{k=1}^{n} \alpha_k(t), \quad t \geq 0, 
\end{equation}   
respectively, and $\zeta_{k+1}(t)$ (for $k=0, 1, 2, \ldots$) is obtained in a recursive way as
$$\zeta_{k+1}(t)=\frac{1}{k+1} \sum_{j=1}^{k} j \eta_j(t) \zeta_{k+1-j}(t), $$
with $\zeta_0(t)=1$ and $\eta_k(t)$ is given by
$$\eta_k(t)=\sum_{j=1}^n \alpha_j(t)(1-\frac{\beta_0}{b_k \beta_k})^k /k. $$
In the particular case that $b_k \beta_k= b \beta$ for all $k$, then $Y(t) \sim \textrm{Gamma}(\sum_{k=1}^n \alpha_k(t), b\beta)$. That is, if $b_k \beta_k= b \beta$ for all $k$, $\{Y(t), t \geq 0\}$ is a gamma process.
\begin{example} We consider a system subject to three degradation processes $\left\{X_1(t), t \geq 0\right\}$, $\left\{X_2(t), t \geq 0\right\}$ and $\left\{X_3(t), t \geq 0\right\}$, respectively. These degradation processes start at random times according to a homogeneous Poisson process with parameter $\lambda=1$. The degradation processes develop according to non-homogeneous gamma process with parameters $\alpha_1=1.1$, $\beta_1=1.1$, $\alpha_2=1.2$, $\beta_2=1.2$ and $\alpha_3=1.3$, $\beta_3=1.3$. Figure \ref{complexgamma} shows these degradation processes and the process $Y(t)=\sum_{j=1}^{3} b_j X_j(t)$ with $b_1=1$, $b_2=0.8$, and $b_3=0.9$. 
	
\begin{figure}
	\begin{center}
		\includegraphics[width=0.45\textwidth]{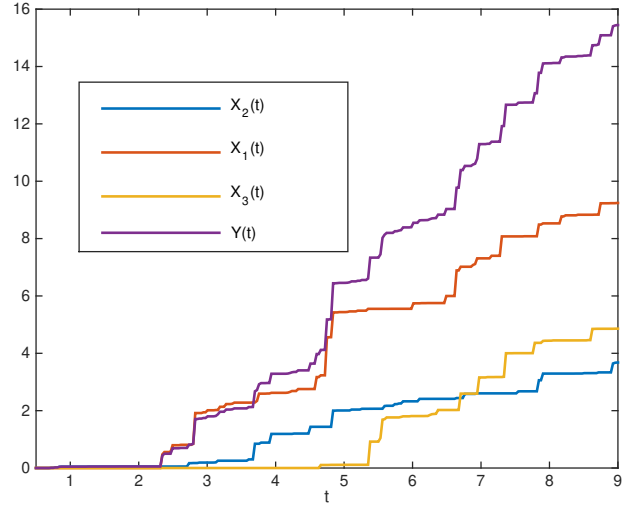}
		\caption{Degradation processes and a linear combination} \label{complexgamma}
	\end{center}
\end{figure}
\end{example}
\subsubsection{First hitting time} To characterise the maintenance scheme of this system, the distribution of the hitting times of the process $\{Y(t), t \geq 0\}$ is obtained. Starting from $Y(0)=0$ and for a fixed degradation level $L$, the first hitting time $\sigma_L$ is defined as the amount of time required for the process $\{Y(t), t \geq 0\}$ to reach the degradation level $L$, that is, 
$$\sigma_{L}= \textrm{inf}(t > 0: Y(t) \ge L). $$ 
The distribution of $\sigma_{L}$ is obtained as 
\begin{align}   
F_{\sigma_L}(t) &= P(Y(t) \geq L) \nonumber \\ 
&= \int_{L}^{\infty} g_{Y(t)}(y)dy \nonumber \\
&= \int_{L}^{\infty} D(t) \sum_{k=0}^{\infty}\frac{\zeta_{k}(t) \beta_0^{-\rho(t)-k}}{\Gamma(\rho(t) + k)} y^{\rho(t) +k -1} e^{-y/\beta_0} dy \nonumber \\ 
&= D(t) \left( \sum_{k=0}^{\infty}\frac{\zeta_{k}(t) \beta_0^{-\rho(t)-k}}{\Gamma(\rho(t) + k)} \int_{L}^{\infty} y^{\rho(t) +k -1} e^{-y/\beta_0} dy \right) \nonumber \\ \label{distributionhittingtime}
&= D(t) \sum_{k=0}^{\infty}\frac{\zeta_{k}(t)}{\Gamma(\rho(t)+ k)} \Gamma_{\mathrm{ui}}(\rho(t)+k, L/\beta_0),
\end{align}   where $\Gamma_{\mathrm{ui}}(\rho(t)+k, L/ \beta_0)$ denotes the upper incomplete gamma function, which is given by
$$\Gamma_{\mathrm{ui}}(\rho(t)+k, L/ \beta_0)=\int_{L/ \beta_0}^{\infty} z^{\rho(t)+k-1}e^{-z}dz. $$
We can link the probability distribution $F_{\sigma_L}(t)$ with the probability distribution of the hitting times for the processes $X_k(t) (k=1=2,...,n$  that compose $Y(t)$. That is, $F_{\sigma_L}(t)$ can be expressed by
$$F_{\sigma_L}(t)=D(t) \sum_{k=0}^{\infty}\zeta_{k}(t)F_{\sigma_{L,k}^*}(t), \quad t \geq 0, $$
where $F_{\sigma_{L,k}^*}(t)$ denotes the distribution of the first hitting time to exceed $L$ for a gamma process with parameters $\rho(t)+k$ and $\beta_0$, where $\rho(t)$ is given by (\ref{beta0D}) and $\beta_0=\min_{1 \leq k \leq n} b_k \beta_k$. 
\subsection{The process of repair cost}
{\color{black} Cost of repairing different effects, such as fatigue cracking and pavement deformation in a pavement network, may be different.} Denote $c_{k,y}$ as the cost of repairing the $k$th defect with deterioration level $y$. We assume in this section that this cost is proportional to the deterioration level, that is, $c_{k,y}=yc_k$, {\color{black} where $c_k$ is the cost of repairing the $k$th defect per unit deterioration level.} We define $U(t)=\sum_{k=1}^n c_k X_k(t)$ as the total repair cost at time $t$. Then $\{U(t), t \ge 0\}$ is the cost growth process and $c_k X_k(t)$ has pdf $f(x;\alpha_k(t), c_k \beta_k)$. The expected value and the variance of $U(t)$ can be obtained by replacing $b_k$ with $c_kb_k$ in Eq. \eqref{Y_Mean} and Eq. \eqref{Y_Variance}, respectively. The pdf of $U(t)=\sum_{k=1}^n c_k X_k(t)$ can be obtained via replacing $b_k$ with $c_k b_k$ in the pdf of $Y$ in Eq. \eqref{Y_distribution}. 

The covariance between $Y(t)$ and $U(t)$ is given by
 \begin{equation}   
\textrm{Cov}(Y(t),U(t))=\sum_{k=1}^n \sum_{j = 1}^n b_k c_{j} \textrm{cov}(X_k(t), X_j(t)).
 \end{equation}   Since $X_k(t)$ for $k=1,2,...$ are independent, $\textrm{cov}(X_k(t), X_j(t))=0$ for $k \neq j$, then
 \begin{equation}   
\textrm{Cov}(Y(t),U(t))= \sum_{k=1}^n b_k c_k \alpha_k(t)  \beta_k^2(t).
 \end{equation}   
In most existing research {\color{black} on maintenance cost of one degradation process $Y(t)$}, once the magnitude of the degradation $Y(t)$ is given, the associated cost of repair may be $c_r Y(t)$, which is proportional to $Y(t)$ (where $c_r$ denotes the cost of repairing a unit of $Y(t)$). {\color{black} However, in our setting, $U(t)$ forms a stochastic process that is not proportional to $Y(t)$. This is because there are many different combinations of $Y_k(t)$ that can be summed up to obtain the same value of $Y(t)$. Correspondingly, the different $Y_k(t)$'s incur different repair cost $c_{k,y}$. As such, for a $Y(t)$ at a given time point $t$, its associated repair cost $U(t)$ is a random variable that does not have a linear correlationship with $Y(t)$.}

The next result gives the distribution of the cost of repair, conditioning that the linear combination of the degradation processes exceeds the pre-specified value $L$.
\begin{lemma} \label{lemacondicion} The conditional probability $f_{U(t)|Y(t)}(y,z)$ is given by
	\begin{align}   
	f_{U(t)|Y(t)}(y,u)
	&=\frac{1}{4\pi^2 g_{Y(t)}(y)} \int_{-\infty}^{\infty}
	\int_{-\infty}^{\infty} \left( \prod_{k=1}^n (1-i(b_k t_1 +c_k t_2)  \beta_k)^{-\alpha_k(t)} \right) e^{-it_1 y-it_2 u} \d t_1\d t_2.
	\end{align}   	\label{ConditionalProbability}
\end{lemma}
\begin{proof} The characteristic function of the bivariate vector $(Y(t),U(t))$ is derived by
	\begin{align}   
	\phi_{Y(t),U(t)}(t_1,t_2)&= \mathbb{E}[\exp(i t_1 Y(t) + i t_2 U(t))] \nonumber\\
	&= \mathbb{E}[\exp(i t_1 \sum_{k=1}^n b_k X_k(t) + i t_2 \sum_{k=1}^n c_k X_k(t))] \nonumber\\
	&= \mathbb{E}[\exp(i \sum_{k=1}^n (b_k t_1 + c_k  t_2) X_k(t)] \nonumber\\
	&= \prod_{k=1}^n \mathbb{E}[\exp(i (b_k t_1 + c_k  t_2) X_k(t)] \nonumber\\
	&=\prod_{k=1}^n \phi_{X_k(t)}(b_k t_1 + c_k t_2).
	\end{align}   Since $\phi_{X_k(t)}(b_k t_1 +c_k  t_2)=(1-i(b_k t_1 +c_k t_2) \beta_k^{-1})^{-\alpha_k(t)}$, we can obtain
	\begin{align}   
	f_{Y(t),U(t)}(y,u)&=\frac{1}{4\pi^2} \int_{-\infty}^{\infty} 
	\int_{-\infty}^{\infty} \phi_{Y(t),U(t)}(t_1,t_2) e^{-it_1 y-it_2 u} \d t_1 \d t_2 \nonumber\\
	&=\frac{1}{4\pi^2} \int_{-\infty}^{\infty}
	\int_{-\infty}^{\infty} \left( \prod_{k=1}^n \phi_{X_k(t)}(b_k t_1 + c_k t_2) \right) e^{-it_1 y-it_2 u} \d t_1 \d t_2 \nonumber\\
	&=\frac{1}{4\pi^2} \int_{-\infty}^{\infty}
	\int_{-\infty}^{\infty} \left(\prod_{k=1}^n (1-i( b_k t_1 +c_k t_2)  \beta_k^{-1})^{-\alpha_k(t)}\right) e^{-i t_1 y-it_2 u} \d t_1 \d t_2.
	\end{align}   Hence, the conditional probability $f_{U(t)|Y(t)}(y,u)$ is given by
	\begin{align}   
	f_{U(t)|Y(t)}(y,u) &= \frac{f_{Y(t),U(t)}(y,u)}{g_{Y(t)}(y)}\nonumber \\
	&=\frac{1}{4\pi^2 g_{Y(t)}(y)} \int_{-\infty}^{\infty}
	\int_{-\infty}^{\infty} \left( \prod_{k=1}^n (1-i(b_k t_1 +c_k t_2) \beta_k)^{-\alpha_k(t)} \right) e^{-it_1 y-it_2 u} \d t_1 \d t_2,
	\end{align}   where $g_{Y(t)}(y)$ is given by (\ref{Y_distribution}).
	This establishes Lemma \ref{ConditionalProbability}. \hfill{$\blacksquare$}
\end{proof}


\subsection{Incorporating random effect}
{\color{black} It is known that random environment may affect the degradation processes of a system.} For example, the deterioration processes of the defects on a pavement network may be affected by covariates such as the weather condition (the amount of rainfall) and traffic loading. If it is possible to collect weather condition data (eg., the amount of rainfall in a time period) and traffic loading data, one may incorporate co-variates in modelling. In addition, we may also consider random effects to account for possible model misspecification and individual unit variability.

\cite{Bagdonavicius} and \cite{Lawless} consider covariates in a gamma process. When incorporating covariates, represented by vector $z$, for example, \cite{Bagdonavicius} incorporate $\alpha_k(t)$ with $\alpha_k(te^{z^{\tau} \delta})$ (where $z^{\tau}$ is the transpose of $z$), \cite{Lawless} replace $\beta_k$ with $\beta_k(z)$, in which $z$ represents covariates and has the effect of rescaling X(t) without changing the shape parameter of its gamma distribution. $\beta_k(z)$ may have a regression function expression such as $\beta_k(z)=\exp(\boldsymbol{\beta' z})$, where $\boldsymbol{\beta'}$ and $\boldsymbol{z}$ are vectors of covariates and regression coefficients, respectively. In the following, we adopt the latter method and assume a degradation process $\{X'_k(t), t \ge 0\}$, which takes both covariates and random effects into consideration. Then, $X'_k(t)$ has density function $f_{\gamma}(x';\alpha_k(t),w_0\beta_{z,k})$, where $w_0$ is a random effect and $\beta_{z,k}$ represents $\beta_k(z)$. One may assume that $w=w_0^{-1}$ has gamma distribution Gamma$(\gamma^{-1},\delta)$ and density function $g_{\gamma^{-1},\delta}(w)=\frac{\gamma^{\delta}}{\Gamma(\delta)} w^{\delta-1}e^{-\gamma w}$; $w$ has mean $\frac{\delta}{\gamma}$ and variance $\sigma^2_z=\frac{\delta}{\gamma^2}$. If $(X'_1, X'_2,\dots,X'_n,w_0)$ has joint density $h(x_1, x_2,\dots,x_n,w)$, then the conditional density of $X'_1, X'_2,\dots,X'_n$ given $w_0=w$, is
 \begin{equation}   
h_0(x_1, x_2,\dots,x_n|w)=
 \frac{h(x_1, x_2,\dots,x_n,w)}{g_{\gamma^{-1},\delta}(w)}.
 \end{equation}   For given weather conditions and traffic loading, one can regard $X_1, X_2,...,X_n$ as independent. That is, $X_1', X_2',...,X_n'$ are conditionally independent given a third event. Then,
\begin{align}   
h(x_1, x_2,\dots,x_n,w)=&h_0(x_1, x_2,\dots,x_n|w)g_{\gamma^{-1},\delta}(w) \nonumber \\
=& g_{\gamma^{-1},\delta}(w) \prod_{k=1}^n h_k(x_k|w). 
\end{align}   Since $h_k(x_k|w)=\frac{(\frac{\beta_{z,k}}{w})^{-\alpha_k(t)}}{\Gamma(\alpha_k (t))}x_k^{\alpha_k(t)-1}e^{- w \frac{x_k}{\beta_{z,k}}}
$, if $(X'_1, X'_2,\dots,X'_n)$ has joint density function $f_0(x_1, x_2,\dots,x_n)$, then
\begin{align}   
f_0(x_1, x_2,\dots,x_n)&=\int_0^{+\infty} g_{\gamma^{-1},\delta}(w) \prod_{k=1}^n h_k(x_k|w) \d w \nonumber \\
&=\int_0^{+\infty} \left(\prod_{k=1}^{n} \frac{(\frac{\beta_{z,k}}{w})^{-\alpha_k(t)}}{\Gamma(\alpha_k (t))}x_k^{\alpha_k(t)-1}\right)e^{-w\sum_{k=1}^n \frac{x_k}{\beta_{z,k}}} g_{\gamma^{-1},\delta}(w) \d w \nonumber \\
&= \prod_{k=1}^{n} \frac{\beta_{z,k}^{-\alpha_k(t)}}{\Gamma(\alpha_k(t))}x_k^{\alpha_k(t)-1}\int_{0}^{\infty} \frac{\gamma^{\delta}}{\Gamma(\delta)}w^{\delta+\rho(t)-1}\exp \left\{-w\left(\gamma+\sum_{k=1}^{n}\frac{x_k}{\beta_{z,k}}\right)\right\} \d w \\
&= \int_{0}^{\infty} \frac{\gamma^{\delta}}{\Gamma(\delta)}w^{\delta+\rho(t)-1}exp \left\{-w\left(\gamma+\sum_{k=1}^{n}\frac{x_k}{\beta_{z,k}}\right)\right\} \d w \left(\prod_{k=1}^{n} \frac{\beta_{z,k}^{-\alpha_k(t)}}{\Gamma(\alpha_k(t))}x_k^{\alpha_k(t)-1}\right) \\
&=\frac{\gamma^{\delta}}{\Gamma(\delta)}\left(\gamma+\sum_{k=1}^{n} \frac{x_k}{\beta_{z,k}}\right)^{\delta+\rho(t)} \Gamma(\delta+\rho(t))\left(\prod_{k=1}^{n} \frac{\beta_{z,k}^{-\alpha_k(t)}}{\Gamma(\alpha_k(t))}x_k^{\alpha_k(t)-1}\right). 
\label{Joint_Dis}
\end{align}   where $\rho(t)=\sum_{k=1}^{n} \alpha_k(t)$. 
\subsubsection{First hitting time}
Next, we compute the first hitting time of the process $\left\{Y(t), t \geq 0\right\}$ to exceed a degradation level $L$. 
Let $$\sigma_L= \textrm{inf}(t > 0: Y(t) \ge L).$$ Then the probability distribution of $F_{\sigma_L}$ is given by
\begin{align}   
F_{\sigma_L}(t)=& \sum_{k=1}^n b_k X_k(t) \geq L
\int_0^{+\infty} g_{\gamma^{-1},\delta}(w) \prod_{k=1}^n h_k(x_k|w) d w d x_1 ... d x_n \nonumber \\ 
=&\sum_{k=1}^n b_k X_k(t) \geq L
\int_0^{+\infty} \left(\prod_{k=1}^{n} \frac{(\frac{\beta_{z,k}}{w})^{-\alpha_k(t)}}{\Gamma(\alpha_k (t))}x_k^{\alpha_k(t)-1}\right)e^{-w\sum_{k=1}^n \frac{x_k}{\beta_{z,k}}} g_{\gamma^{-1},\delta}(w) d w d x_1 ... d x_n \nonumber \\
=&\int_0^{+\infty}\left(\sum_{k=1}^n b_k X_k(t)\geq L \left(\prod_{k=1}^{n} \frac{(\frac{\beta_{z,k}}{w})^{-\alpha_k(t)}}{\Gamma(\alpha_k (t))}x_k^{\alpha_k(t)-1}\right)e^{-w\sum_{k=1}^n \frac{x_k}{\beta_{z,k}}} d x_1 ... d x_n \right) g_{\gamma^{-1},\delta}(w) d w. 
\end{align}   
According to \cite{Moschopoulos}, we have $$\sum_{k=1}^n b_k X_k(t)<L \left(\prod_{k=1}^{n} \frac{(\frac{\beta_{z,k}}{w})^{-\alpha_k(t)}}{\Gamma(\alpha_k (t))}x_k^{\alpha_k(t)-1}\right)e^{-w\sum_{k=1}^n \frac{x_k}{\beta_{z,k}}} g_{\gamma^{-1},\delta}(w) d w d x_1 ... d x_n
=\int_{L}^{\infty} g_{Y(t)}'(x) \d x, $$
where $g_{Y(t)}'(x)$ is obtained following the same reasoning as in (\ref{Y_distribution}), that is, 
$$g_{Y(t)}'(x)=D_z(t) \sum_{k=0}^{\infty}\frac{\zeta_{z,k}(t) (\beta_{z,0}/w)^{-\rho_z(t)-k}}{\Gamma(\rho_z(t) + k)} x^{\rho_z(t) +k -1} e^{-w x/\beta_{z,0}}, $$ and $D_z(t)$, $\zeta_{z,k}(t)$, $\rho_z(t)$ and $\beta_{z,0}$ are obtained by replacing $\beta_k$ with $\beta_{z,k}$ in the definitions of $D(t)$, $\zeta_k(t)$, $\rho(t)$ and $\beta_0$, respectively. 

Finally, we obtain
\begin{align}   
P(Y(t) \geq L)
=&\int_0^{\infty} \int_{L}^{\infty} g_{Y(t)}'(y) g_{\gamma^{-1},\delta}(w) \d y \d w 
 \nonumber \\
=& D_z(t) \gamma^{\delta}\sum_{k=0}^{\infty}\frac{\zeta_{z,k}(t) \beta_{z,0}^{-\rho_z(t)-k}}{B(\rho_z(t)+k,\delta)} \int_{L}^{\infty} x^{\rho_z(t) +k -1} \left(\gamma+\frac{x}{\beta_{z,0}}\right)^{\rho_z(t)+k+ \delta}\d x \nonumber \\
 =&D_z(t) (\beta_{z,0} \gamma)^{\delta}\sum_{k=0}^{\infty}\frac{\zeta_{z,k}(t)}{B(\rho_z(t)+k,\delta)} \int_{L}^{\infty} x^{\rho_z(t) +k -1} \left(\beta_{z,0}\gamma+x\right)^{\rho_z(t)+k+ \delta}\d x, 
\end{align} 
where $B(x,y)$ denotes the function beta given by 
$$B(x,y)=\frac{\Gamma(x)\Gamma(y)}{\Gamma(x+y)}.$$
\section{Maintenance Policies}
In the reliability literature, there are many models describing the effectiveness of a maintenance activity. Such models include modification of intensity models \cite{Doyen}, \cite{Wu2}, reduction of age models \cite{Kijima},\cite{Doyen},\cite{Wu2}, geometric processes \cite{Lam},\cite{Wu3},\cite{Wu1}, etc. For a system like a section of pavement, maintenance may remove all of the defects, the degradation processes of the defects may therefore stop. After maintenance, new defects may develop in a faster manner than before. The effectiveness of such maintenance may be modelled by the geometric process.

The geometric process describes a process in which the lifetime of a system becomes shorter after each maintenance. Its definition is given by \cite{Lam}, and it is shown below.
\begin{definition} \cite{Lam}
	Given a sequence of non-negative random variables $\{X_j,j=1,2, \dots\}$, if they are independent and the cdf of $X_j$ is given by $F(a^{j-1}x)$ for $j=1,2, \dots$, where $a$ is a positive constant, then $\{X_j,j=1,2,\cdots\}$ is called a geometric process (GP). 
	\label{GP}
\end{definition}
The parameter $a$ in the GP plays an important role. The lifetime described by $F(a^{j-1}x)$ with a larger $a$ is shorter than that described by $F(a^{j-1}x)$ with a smaller $a$ with $j=1, 2, \ldots$.
\begin{itemize}
	\item If $a>1$, then $\{X_j, j=1, 2, \cdots\}$ is stochastically decreasing.
	\item If $a<1$, then $\{X_j, j=1, 2, \cdots\}$ is stochastically increasing.
	\item If $a=1$, then $\{X_j, j=1, 2, \cdots\}$ is a renewal process.
	\item If $\{X_j, j=1,2,\dots\}$ is a GP and $X_1$ follows the gamma distribution, then the shape parameter of $X_j$ for $j=2,3,\dots$ remains the same as that of $X_1$ but its scale parameter changes.
\end{itemize}
GP has been used extensively in the reliability literature to implement the effect of imperfect repairs on a repairable system (see \cite{Castro}, \cite{Wang}, \cite{Wu3}, \cite{Wu1}, among others). 

In addition to the assumptions listed in Section \ref{Assumptions}, we make the following assumptions. 
\begin{enumerate}[label=A\arabic*).,start=8]
\item Immediately after a repair, the system resets {\it its} age to 0, at which there are no defects in the system. 
\item The initiation of the defects after the $j$-th imperfect repair follows a homogeneous Poisson process with parameters $\lambda/a_1(T)^{j-1}$ with $a_1(T)>0$ and $a_1(T)$ being a non-decreasing function in $T$ for $j=1, 2, \ldots$. 
\item After the $j$-th repair and after the arrival of the $k$-th defect, the $k$-th defect grows according to a gamma process with shape parameter $\alpha_k(t)$ and scale parameter $a_2(T)^{j-1} \beta_k$ with 
$a_2(T)(>0)$ being $a_2(T)$ an increasing function in $T$ for $j=1, 2, \ldots$. 
\item Each inspection implies a cost of $c_I$ monetary units, $c_{k,y}$ corresponds to the variable cost of repairing the $k$-th defect with degradation level equals to $y$ and $c_{f,k}$ corresponds to the fixed cost of repairing the $k$-th defect. Furthermore, if in an inspection time the ``overall degradation'' of the system given by (\ref{overalldegradation}) exceeds the threshold $L$, an additional cost of $c_F$ monetary units is incurred. The cost of the replacement at time $NT$ is equal to $c_R$. 
\end{enumerate}
We explain assumptions A9) and A10), respectively, in the following.
\begin{itemize}
\item Assumption A9) implies that the defect arrival rate relates to the inspection interval $T$, which reflects the case that $a_1(T)$ becomes bigger and the system tends to deteriorate  faster for large $T$ than for small $T$. 
\item Assumption A10) implies that the degradation rate increases with the number of imperfect repairs performed on the system. We denote by $\left\{Y^{*}_j(t), t \geq 0\right\}$ the ``overall'' degradation of the maintained system after the $j$-th repair, and denote
	\begin{equation}    \label{overalldegradation}
	Y^*_j(t)=\sum_{k=1}^{n}b_kX_{k,j}(t), \quad 0 \leq t \leq T, 
 \end{equation}   
where $\left\{X_{k,j}(t), t \geq 0\right\}$ stands for a gamma process with parameters $\alpha_k$ and $\beta_k a_2(T)^{j-1}$. Similar to the derivation process shown in previous section, we can compute the first hitting time to exceed the threshold $L$ for the process (\ref{overalldegradation}) following the same reasoning as in (\ref{Y_distribution}) replacing $\beta_k$ by $\beta_k a_2(T)^{j-1}$. That is, 
$$\sigma_L^{(j)}=\inf \left\{t \geq 0: Y^*_j(t) \geq L\right\}, $$
	and we denote by $F_{\sigma_L}^j$ the distribution of $\sigma_L^{(j)}$. 
\end{itemize}

The problem is to determine the time between inspections and the number of inspections that minimise an objective cost function. The optimisation problem is formulated in terms of the expected cost rate per unit time.

By a replacement cycle, we mean the time between two successive replacements of the system. In this paper, the total replacement cycle is equal to $NT$. Let $\mathrm{Q}_0(N,T)$ be the expected rate of the total cost in a replacement cycle. Then we obtain
\begin{align}    \nonumber
 \mathrm{Q}_0(N,T) = & \frac{1}{NT}\sum_{j=1}^N \left[c_I + \sum_{k=1}^n \frac{(a_1(T))^{j-1}}{\lambda}\left(c_{f,k} + \int_{0}^{\infty}c_{k,y} f(y; \alpha_k(T),\beta_k a_2(T)^{j-1}) \d y \right) \right. \nonumber \\ 
& \left. + c_F \frac{(a_1(T))^{j-1}}{\lambda} F_{\sigma_L}^{j}(T)\right]+\frac{c_R}{NT},
\label{ExpectedDegradationCostValue}
\end{align}   where $f(y; \alpha_k(T),\beta_k a_2(T)^{j-1})$ is given by (\ref{gammapdf}) and $F_{\sigma_L}^j(T)$ is given by (\ref{distributionhittingtime}) replacing $\beta_k$ by $\beta_ka_2(T)^{j-1}$. The expected variable cost per unit time in a replacement cycle is given by
 \begin{equation}    \label{variablecost}
\mathrm{CV}(N,T)=\frac{1}{NT}\sum_{j=1}^{N} \sum_{k=1}^{n}\frac{a_1(T)^{j-1}}{\lambda}\int_{0}^{\infty}c_{k,y} f(y; \alpha_k(T),\beta_k a_2(T)^{j-1}) \d y.
 \end{equation}   
The optimization problem is formulated as 
 \begin{equation}    \label{optimisationproblem}
\mathrm{Q}_0(N_{opt},T_{opt})=\min_{\substack{N=1,2,... \\ T>0}} \mathrm{Q}_0(N,T). 
 \end{equation}    

\subsection{Special cases}
In this section, we discuss $\mathrm{Q}_0(N,T)$ and $\mathrm{CV}(N,T)$ under special cases of $c_{k,y}$, $a_1(T)$, $a_2(T)$, and $\alpha_k(T)$, respectively.
\subsubsection{Special cases of $c_{k,y}$}
Different scenarios can be envisaged depending on the variable cost function  $c_{k,y}$. 
\begin{itemize}
\item If $c_{k,y}=c_k$, then the expected variable cost rate (\ref{variablecost}) in a renewal cycle is given by
 \begin{align*}
\mathrm{CV}(N,T)&=\frac{1}{NT} \sum_{j=1}^{N} \sum_{k=1}^{n}\frac{a_1(T)^{j-1}}{\lambda}\int_{0}^{\infty}c_k f(y; \alpha_k(T),\beta_k a_2(T)^{j-1}) \d y \nonumber \\
&= \frac{1}{NT} \sum_{j=1}^{N} \sum_{k=1}^{n}\frac{a_1(T)^{j-1}c_k}{\lambda}. 
\end{align*}   \item If $c_{k,y}$ is directly proportional to the degradation level of the $k$-th defect in the inspection time, that is, $c_{k,y}=y c_k$, then the expected variable cost (\ref{variablecost}) is equal to
 \begin{align*}
\mathrm{CV}(N,T) &= \frac{1}{NT} \sum_{j=1}^{N} \sum_{k=1}^{n}\frac{a_1(T)^{j-1}}{\lambda}\int_{0}^{\infty}c_k y f(y; \alpha_k(T),\beta_k a_2(T)^{j-1}) \d y \\
&= \frac{1}{NT} \sum_{j=1}^{N} \sum_{k=1}^{n}\frac{a_1(T)^{j-1}}{\lambda} c_k\alpha_k(T) a_2(T)^{j-1}\beta_k \\
&= \frac{1}{NT}\left( \frac{(a_1(T)a_2(T))^N-1}{a_1(T)a_2(T)-1} \right) \sum_{k=1}^{n} \frac{c_k\alpha_k(T)\beta_k}{\lambda}.
\end{align*}   \item If $c_{k,y}$ is directly proportional to the square of the degradation level of the $k$-th defect in the inspection time, that is, $c_{k,y}=y^2 c_k$. In this case, the repair cost may relate to the area of a defect, see \cite{Noortwijk}, for example.  (\ref{variablecost}) is given by
 \begin{align*}
\mathrm{CV}(N,T) &= \frac{1}{NT} \sum_{j=1}^{N} \sum_{k=1}^{n}\frac{a_1(T)^{j-1}}{\lambda}\int_{0}^{\infty}c_k y^2 f(y; \alpha_k(T),\beta_k a_2(T)^{j-1}) \d y \\
&= \frac{1}{NT} \sum_{j=1}^{N} \sum_{k=1}^{n}\frac{a_1(T)^{j-1}}{\lambda} c_k \left(\mathrm{Var}(X_{k,j}(T))+(\mathbb{E}(X_{k,j}(T)))^2\right) \\
&= \frac{1}{NT} \sum_{j=1}^{N} \sum_{k=1}^{n}\frac{a_1(T)^{j-1}}{\lambda} c_k \left(\beta_k^2 \alpha_k(T) a_2(T)^{2j-2} + \beta_k^2 \alpha_k(T)^2 a_2(T)^{2j-2} \right) \\
&= \frac{1}{NT} \left(\frac{(a_1(T)a_2(T)^2)^N-1}{a_1(T)a_2(T)^2-1}\right) \sum_{k=1}^{n} \frac{c_k\beta_k^2(\alpha_k(T)+\alpha_k(T)^2)}{\lambda}.
\end{align*}  
\end{itemize}

\subsubsection{Special cases of $a_1(T)$, $a_2(T)$, and $\alpha_k(T)$}
The analysis of the monotonicity of $\mathrm{Q}_0(N,T)$ is quite tricky. To analyse it, some particular conditions are imposed. We assume that $a_1(T)=a_1$, $a_2(T)=a_2$, $\alpha_k(T)=\alpha_kT$, and $c_{k,y}=yc_k$, 
$\mathrm{Q}_0(N,T)$ given by (\ref{ExpectedDegradationCostValue}) is then reduced to
\begin{align}    \label{expectedreduced}
\mathrm{Q}_0(N,T) &= \frac{c_I}{T} +\frac{c_R}{NT} +\frac{(a_1^N-1)}{(a_1-1)\lambda NT}\sum_{k=1}^{n}c_{f,k} \\
&+ \frac{(a_1^N a_2^N-1)}{\lambda N(a_1a_2-1)}\sum_{k=1}^{n}c_k\alpha_k \beta_k +\frac{c_F}{ \lambda N T} \sum_{j=1}^{N}a_1^{j-1}F_{\sigma_L}^{(j)}(T). \nonumber
\end{align}   We suppose that $N$ is constant and $T$ is variable on $(0, \infty)$. A necessary condition that a finite $T^{*}$ minimises $\mathrm{Q}_0(N,T)$ given by (\ref{expectedreduced}) is that it satisfies
$$ \sum_{j=1}^{N} a_1^{j-1}\left(f_{\sigma_L}^{(j)}(T)T-F_{\sigma_L}^{(j)}(T)\right) =\frac{\lambda}{c_F}\left(Nc_I+\frac{a_1^N-1}{\lambda(a_1-1)} \sum_{k=1}^{n} c_{f,k}+c_R\right). $$
Next, we suppose that $T$ is constant. Then a necessary condition that there exists a finite a unique $N^*$ minimizing $\mathrm{Q}_0(N,T)$ is that $N^*$ satisfies 
$$\mathrm{Q}_0(N+1,T) \geq \mathrm{Q}_0(N,T),$$
and
$$\mathrm{Q}_0(N,T) \geq \mathrm{Q}_0(N-1,T).$$
We get that
 \begin{align*}
\mathrm{Q}_0(N+1,T)-\mathrm{Q}_0(N,T) =&\frac{1}{\lambda T (a_1-1)}\sum_{k=1}^{n}c_{f,k}\frac{N(a_1^{N+1}-a_1^N)-a_1^N+1}{N(N+1)} \\ &	+\frac{\sum_{k=1}^{n}c_k\alpha_k\beta_k}{\lambda(a_1a_2-1)}\frac{N(a_1^{N+1}a_2^{N+1}-a_1^Na_2^N)-a_1^Na_2^N+1}{N(N+1)} \\
& -\frac{c_R}{N(N+1)T}+ c_F \frac{\sum_{j=1}^{N}a_1^{N}F_{\sigma_L}^{(N+1)}(T)-a_1^{j-1}F_{\sigma_L}^(j)(T)}{\lambda N(N+1)T}. 
\end{align*}   Hence, for fixed $T$, $\mathrm{Q}_0((N+1),T)-\mathrm{Q}_0(N,T) \geq 0$ if and only if
$$c_R < D(N,T), $$
where
 \begin{align*}
D(N,T) &= \frac{1}{\lambda}\sum_{k=1}^{n}c_{f,k}\frac{N(a_1^{N+1}-a_1^N)-a_1^N+1}{(a_1-1)} \\
&+ T \sum_{k=1}^{n}c_k\alpha_k\beta_k \frac{N(a_1^{N+1}a_2^{N+1}-a_1^Na_2^N)-a_1^Na_2^N+1}{\lambda(a_1a_2-1)} \\
& + \frac{c_F }{\lambda}\left(\sum_{j=1}^{N}a_1^N F_{\sigma_L}^{(N+1)}(T)-a_1^{j-1} F_{\sigma_L}^{(j)}(T)\right). 
\end{align*}   We get that, if $a_1>2$, then $D(N,T)$ is non decreasing in $N$. Therefore, if
$$ c_R< D(1,T), $$
then $ c_R<D(N,T)$ for all $N$. 
We get that
 \begin{align*}
D(1,T) &= \frac{(a_1-1)}{\lambda}\sum_{k=1}^{n} c_{f,k}+\frac{T(a_1a_2-1)}{\lambda} \sum_{k=1}^{n}c_k\alpha_k\beta_k + \frac{ c_F }{\lambda} \left(a_1F_{\sigma_L}^{(2)}(T)-F_{\sigma_L}^{(1)}(T)\right).
\end{align*}   Hence, if $a_1>2$ and
$$c_R<\frac{(a_1-1)}{\lambda} \sum_{k=1}^{n}c_{f,k}, $$
then $\mathrm{Q}_0(N,T)$ is increasing in $N$. 

An economic constraint is introduced in the optimisation problem formulated in (\ref{optimisationproblem}) to limit the variable cost in a replacement cycle. The introduction of constraints in the search of the optimal maintenance strategy is not new in the literature. For example, \cite{Aven1} and \cite{Aven2} introduced constraints related to the system safety in an optimisation problem. In this paper, the constraint imposed in the optimisation is economic and it is related to the expected variable cost imposing that this expected variables cost cannot exceed a threshold K. 

Let $\Omega$ be the set of pairs $(N,T)$ such that $\mathrm{CV}(N,T)\leq K$, that is,
 \begin{equation}    \label{omega}
\Omega=\left\{(N,T): N=1, 2, \ldots, T >0 \, \, \mbox{subject to} \, \, \mathrm{CV}(N,T) \leq K \right\}, 
 \end{equation}   and the optimisation problem is formulated in terms of the economic constraint as
 \begin{equation}    \label{optimisationproblemconstraint}
\mathrm{Q}_0^*(T_{opt},N_{opt})=\inf \left\{\mathrm{Q}_0(N,T): (N,T) \in \Omega \right\}. 
 \end{equation}   To analyse the optimisation problem given by (\ref{optimisationproblemconstraint}), the monotonicity of the function $\mathrm{CV}(N,T)$ is studied. 

 \subsection{Economic constraint analysis}
We analyse the monotonicity of $\mathrm{CV}(N,T)$ in the two variables $N$ and $T$ and assume that $c_{k,y}=y f(y;\alpha_k(T),\beta)$ (i.e., variable cost proportional to the degradation level) and $a_2(T) > 1$ and $a_1(T)>1$ for all $T$. 
\begin{lemma} \label{lemma1} If $\alpha_k(T)$ is convex in $T$ for all $k$ with $\alpha_k(0)=0$ and $\alpha_k'(0)<\infty$, then 
\begin{itemize}
\item $\mathrm{CV}(N,T)$ is increasing in $T$ for fixed $N$, and  
\item $\mathrm{CV}(N,T)$ is increasing in $N$ for fixed $T$. 
\end{itemize}
\end{lemma}
\begin{proof} The expected variable cost rate is given by
\begin{align}   
\mathrm{CV}(N,T)&=\frac{1}{N} \left(\sum_{k=1}^{n} \frac{c_k\beta_k \alpha_k(T)/T}{\lambda}\right) \left( \sum_{j=0}^{N-1}(a_1(T)a_2(T))^j\right)\nonumber \\
&=\frac{1}{N} \frac{(a_1(T)a_2(T))^N-1}{a_1(T)a_2(T)-1}\left(\sum_{k=1}^{n} \frac{c_k\beta_k \alpha_k(T)/T}{\lambda}\right). 
\end{align}   The function $\alpha_k(T)/T$ is increasing in $T$ as consequence of the convexity of $\alpha_k(T)$ along with $\alpha_k(0)=0$ and $\alpha_k'(0)<\infty$. On the one hand, since $a_1(T)$ and $a_2(T)$ are increasing in $T$, $\mathrm{CV}(N,T)$ is increasing in $T$. On the other hand, the function
$$g(N)=\frac{1}{N}\sum_{j=1}^{N-1} (a_1(T)a_2(T))^j, $$
is increasing in $N$ since
$$g(N+1)-g(N) =\frac{\sum_{j=0}^{N-1}((a_1(T)a_2(T))^N-(a_1(T)a_2(T))^j)}{N(N+1)}, $$
and $a_1(T)a_2(T) > 1$, hence $\mathrm{CV}(N,T)$ is increasing in $N$.  This establishes Lemma 2.
\hfill{$\blacksquare$}
\end{proof}

The first consequence of Lemma \ref{lemma1} is that the condition
 \begin{equation}    \label{condicion1} \frac{1}{\lambda}
\sum_{k=1}^{n} c_k\beta_k \lim_{T \rightarrow 0} \frac{\alpha_k(T)}{T} \leq K, 
 \end{equation}   has to be imposed. If inequality \eqref{condicion1} is not fulfilled, then $\Omega=\emptyset$. 
On the other hand, if
 \begin{equation}   
\lim_{T \rightarrow \infty} \lim_{N \rightarrow \infty} \mathrm{CV}(N,T) \leq K, 
 \end{equation}   then $\Omega=\left\{{\color{black}(N,T):}  T > 0, \, N=1, 2, \ldots \right\}, $ and the optimisation problem in (\ref{optimisationproblemconstraint}) is reduced to the optimisation problem in (\ref{optimisationproblem}). 
Hence, to deal with the optimisation problem with constraints, we assume that the following inequality
 \begin{equation}    \label{inequalityconstraint}
\frac{1}{\lambda} \sum_{k=1}^{n} c_k\beta_k \lim_{T \rightarrow 0} \frac{\alpha_k(T)}{T} \leq K < \lim_{T \rightarrow \infty} \lim_{N \rightarrow \infty} \mathrm{CV}(N,T), 
 \end{equation}   is fulfilled. If (\ref{inequalityconstraint}) is fulfilled, we denote
$$N_1=\inf \left\{N: \lim_{T \rightarrow \infty} \mathrm{CV}(N,T) > K \right\}, $$
and 
$$N_2=\inf \left\{N: \lim_{T \rightarrow 0} \mathrm{CV}(N,T) > K \right\}. $$
We obtain that $N_1 \leq N_2$. 

If $N^*$ is fixed such that $N_1\leq N^{*} \leq N_2$, we denote $T_N^*$ as the root of the equation
$$ \mathrm{CV}(N^{*},T_N^*) = K, $$ 
and the set $\Omega$ given in (\ref{omega}) is therefore equal to
 \begin{equation}   
\Omega=\left\{(N,T): N=1, 2, \ldots, N_1-1 \right\} \cup \left\{(N,T): N=N_1, N_1+1, \ldots, N_2-1, \, T \leq T_N^* \right\}. 
 \end{equation}   

\section{Numerical examples}
We consider a system subject to three different defects, all of which start at random times, following a homogeneous Poisson process with rate $\lambda=1$ defects per unit time. The degradation process of the three defects is modelled using a nonhomogeneous gamma processes with shape parameters $\alpha_k(t)=\alpha_k t^{\xi_k}$ with $\xi_k=2$, $\alpha_1=1$, $\alpha_2=1$, $\alpha_3=1$ and scale parameters $\beta_1=1$, $\beta_2=2$ and $\beta_3=3$, respectively. The random effect $w_0$ is modelled with $w=w_0^{-1}$, where $w$ follows a gamma distribution $Gamma(1,2)$. 

The overall degradation process of the system $Y$ is a combination linear of the three processes
$$Y=0.2 X_1+0.7 X_2+0.4 X_3, $$
and we assume that the system fails when the degradation level of $Y$ exceeds the failure threshold $L=20$. 
Imperfect repairs are performed on the system every $T$ time units and the effect of these imperfect repairs is modelled by a geometric process with parameters $a_1(T)=1.1(1.2-0.2\exp(-T))$ for the time between arrivals and $a_2(T)=1.15(1.2-0.2\exp(-T))$ for the effect of the imperfect repairs on the degradation rate of the defects. Each inspection involves a cost of $c_I=0.05$ monetary units. Each repair involves a fixed cost of $c_{f,1}=2$ monetary units for the first defect, $c_{f,2}=2$ monetary units for the second defect and $c_{f,3}=2$ monetary units for the third defect. Each repair involves also a variable cost depending on the degradation of the defect. The variable cost is given by $c_{1,y}=7y$, $c_{2,y}=7 y$ and $c_{3,y}=7 y$ on the three defects, respectively, where $y$ denotes the degradation of the defect in the time of the repair. If the overall degradation of the system exceeds $L=20$ in the repair time, an additional cost of $c_F=100$ monetary units is incurred. A complete replacement of the system by a new one is performed at the time of the $N$-th imperfect repair with a cost of $c_R=1000$ monetary units. 

Figure \ref{newcost} shows the expected cost per unit time $\mathrm{Q}_0(N,T)$ versus $N$ and $T$. This graphic is obtained by simulation with 10 values for $T$ from 1 to 7, $N$ from 1 to 2 and 3000 simulations in each point. 

By inspection, the minimal value of $\mathrm{Q}_0(N,T)$ are obtained for $T_{opt}=1.9474$ and $N_{opt}=3$ with and optimal expected cost rate of $\mathrm{Q}_0(N_{opt},T_{opt})=332.6066$ monetary units per unit time. 
\begin{figure}
	\begin{center}
		\includegraphics[width=0.5\textwidth]{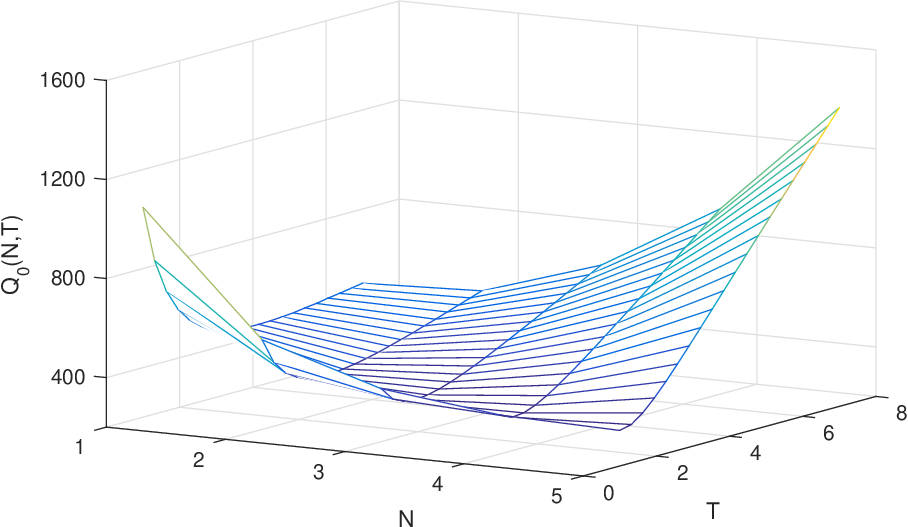}
		\caption{Expected cost $\mathrm{Q}_0(N,T)$ versus $N$ and $T$.} \label{newcost}
	\end{center}
\end{figure}
The economic safety constraint is introduced in this problem and it is dependent on the variable cost given by 
 \begin{equation}    \label{safetyconstraint1}
	\mathrm{CV}(N,T)=\frac{1}{\lambda N} \sum_{k=1}^{n}c_k\alpha_k \beta_k T \sum_{j=0}^{N-1} (a_1(T)a_2(T))^j. 
 \end{equation}   For fixed $N$, the function given by (\ref{safetyconstraint1}) is non-decreasing in $T$. For fixed $T$, we get that
$$\mathrm{CV}(N+1,T)-\mathrm{CV}(N,T)=\sum_{k=1}^{n}\frac{c_k\alpha_k\beta_k T}{\lambda}\sum_{j=0}^{N-1}\frac{\left((a_1(T)a_2(T))^N-a_1(T)a_2(T)^j\right)}{N(N+1)},$$
is positive. Figure \ref{newvariablecost} shows the economic safety constraint versus $T$ and $N$.  
\begin{figure}
	\begin{center}
		\includegraphics[width=0.5\textwidth]{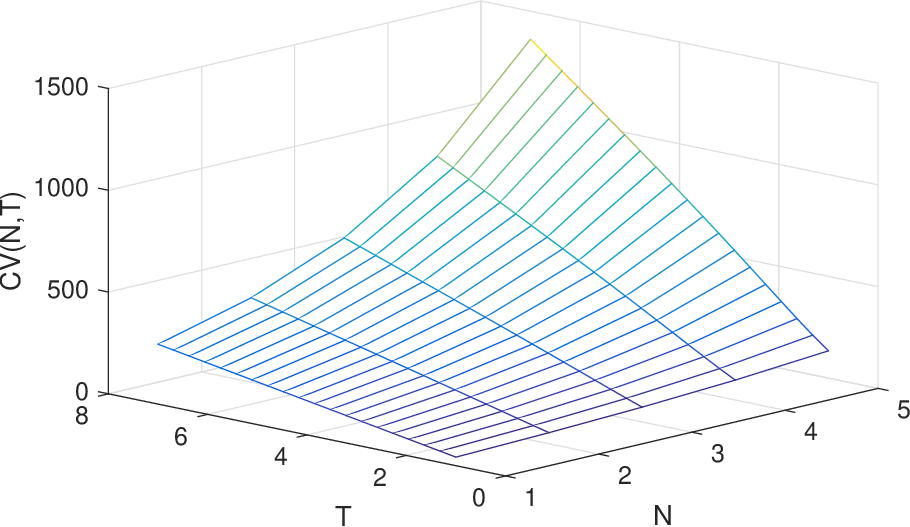}
		\caption{Variable cost $\mathrm{CV}(N,T)$ versus $N$ and $T$.} \label{newvariablecost}
	\end{center}
\end{figure}
As we visually can check, the variable cost is non-decreasing in $N$ for fixed $T$ and non-decreasing in $T$ for fixed $N$. 

We assume that the variable cost cannot exceed the threshold $K=130$ monetary units, that is, the optimization of $\mathrm{Q}_0(N,T)$ given by (\ref{ExpectedDegradationCostValue}) is performed on the set $\Omega_1$, where
$$\Omega_1=\left\{(N,T) \;\;  \mbox{such that} \;\; \, \mathrm{CV}(N,T) \leq 13
0\right\}.$$
Inequality (\ref{inequalityconstraint}) is fulfilled since
$$\lim_{T \rightarrow 0} \mathrm{CV}(1,T)=\lim_{T \rightarrow 0} \frac{1}{\lambda}\sum_{k=1}^{n} c_k\beta_k\alpha_k T^{\xi_k-1} =0, $$
and, therefore, $\lim_{T \rightarrow 0} \mathrm{CV}(1,T) \leq K$, and
$$\lim_{N \rightarrow \infty} \lim_{T \rightarrow \infty}  \mathrm{CV}(N,T) = \infty,  $$
hence inequality (\ref{inequalityconstraint}) is fulfilled. 

Figure \ref{varyN} shows the value of $\mathrm{CV}(N,T)$ for $N \leq 10$. 
\begin{figure}[h!]
	\begin{center}
		\includegraphics[width=0.5\textwidth]{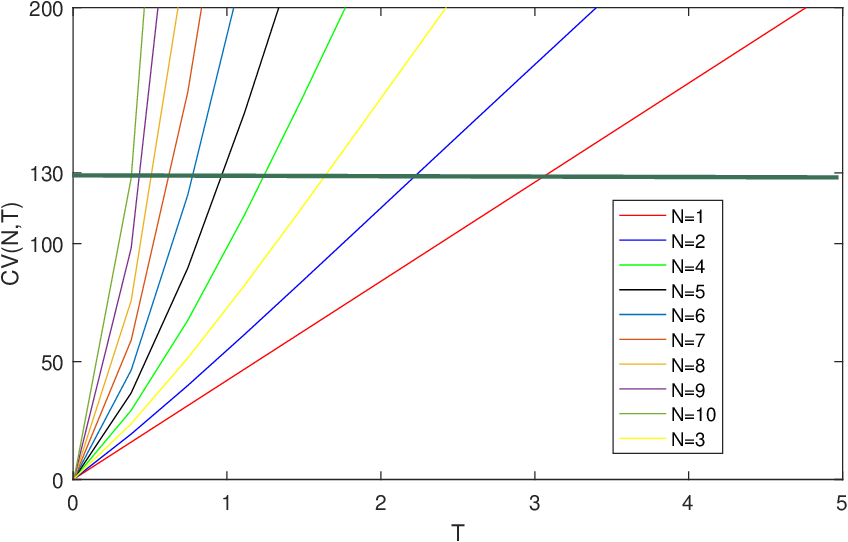}
		\caption{Variable cost $\mathrm{CV}(N,T)$ versus $T$.} \label{varyN}
	\end{center}
\end{figure}
The set of the points $\Omega$ that fulfils the economic constraint is given by
$$\Omega_1=\left\{(N,T), \, N\geq 1; \quad T \leq T_N^*\right\},  $$
where $T_N^*$ is the root of $\mathrm{CV}(N,T)=K$. 

The point in which the global minimum is obtained in the unconstrained problem (that is, $T_{opt}=1.9474$ and $N_{opt}=3$) presents a variable cost equals to $\mathrm{CV}(N_{opt}, T_{opt})=147.8725$ monetary units per unit time what it implies that it is not an optimal solution for the constrained problem. 

Figure \ref{varyN2} shows the values for the expected cost rate $\mathrm{Q}_0(N,T)$ for $T \leq T_N^*$, that is, 
 the expected cost rate $\mathrm{Q}_0(N,T)$ in the subset $\Omega$. The minimum of these function is reached at point $N_{opt}=4$ and $T_{opt}=1.1137$ with an expected cost rate equals to $\mathrm{Q}_0(T_{opt}, N_{opt})=344.4153$ monetary units per unit time. 
  
	\begin{figure}[h!]
	\begin{center}
		\includegraphics[width=0.5\textwidth]{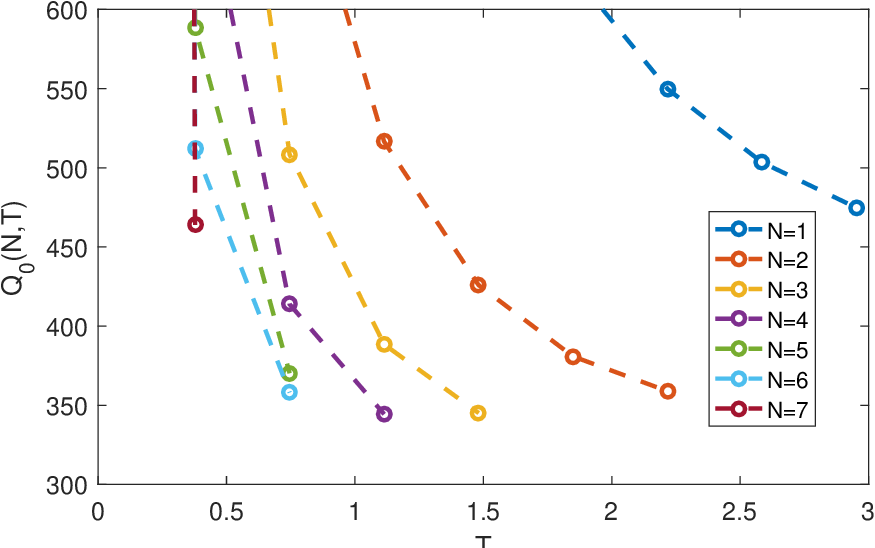}
		\caption{Expected cost rate $\mathrm{Q}_0(N,T)$ versus $T$ in $\Omega$.} \label{varyN2}
	\end{center}
\end{figure}


\section{Discussion}
This paper discussed the scenario where a linear combination of degradation processes was studied. {\color{black} Below we discuss the assumption of the degradation processes, the random environment, and the effectiveness of repair.} 

{\color{black} 

	\subsection{Degradation process.} The preceding sections assume that $X_k(t)$ follows the gamma process.  Certainly, one may choose the degradation process of $X_k(t)$ based on the real applications: for example, in the case of the example investigated in this paper, the propagation process of a fatigue crack evolves monotonically only in one direction, the gamma process is a good choice. Methodologically, however, $X_k(t)$ may be assumed to follow any other process, such as the Wiener process \cite{Sun}, the inverse Gaussian process \cite{Chen} and the Ornstein-Uhlenbeck process \cite{Deng}. The probability distribution of $\sum_{k=1}^n X_k(t)$ can be easily derived if $X_k(t) (k=1,2,...,n)$ follow  Wiener processes. In some case, a closed form of the distribution of $\sum_{k=1}^n X_k(t)$ may not be easily found and therefore numerical methods may be pursued. 
	
	One may also assume that $X_k(t)$ may follow different degradation processes, for example, on different $k$'s, some $X_k(t)$'s follow gamma processes and others follow Wiener processes. 
\subsection{Incorporation of dynamic environments.} The system considered in this paper is operated under a random environment. In addition to the method that incorporate the random environment with the random effect method, one may also use other methods, for example, one may consider the effect of the dynamic environment on the system as external shocks using Poisson processes \cite{Yang}, or as other stochastic processes, including continuous-time Markov chain process \cite{Bian}, and Semi-Markov process \cite{Khaorufeh}. The reader is referred to \cite{Peng} for a discussion in detail.
	\subsection{Imperfect repair.} In this paper, we consider the effectiveness of repair as imperfect. The justification is as follows. If we consider a pavement network, all defects, such as fatigue cracking and pavement deformation, disappear after repair. This does not suggest the pavement network is repaired as good as new (i.e., perfect repair) or as bad as old (i.e., minimal repair). Instead, it is more reasonable to assume that the repair is imperfect. In the literature, many methods that model the effectiveness of imperfect maintenance have been developed (see the Introduction section in \cite{Wu2}, for example). For simplicity, this paper uses the geometric process introduced in \cite{Lam}. Of course, one may other models such as the age-modification models \cite{Kijima},\cite{Doyen}, under which the optimisation process becomes much more complicated.
	\subsection{ Maintenance policy based on the cost process.} Since $U(t)$ forms a stochastic process, one may develop a maintenance policy based on the cost process. That is, once the cost process reaches a threshold, maintenance on the combined degradation process $Y(t)$ is carried out. Hence, intriguing questions may include optimisation of maintenance intervals, for example.
}
\subsection{Rethinking of the assumptions}
The above sections assumes the defect inter-occurrence times to be {\it exchangeable} and to exhibit the {\it lack of memory} property. Nevertheless, both properties may be violated in the real world. If so, one may assume that the defect inter-occurrence times follow a non-homogeneous Poisson process, for example.

\subsection{A $r$-out-of-$n$ case}
\label{r_out_n_Combination}
In Section \ref{Linear_Combination}, we discussed the case when the sum of the deterioration levels is monitored. In practice, another scenario may be to monitor $r$-out-of-$n$ deterioration processes. That is, if $k$-out-of-$n$ deterioration levels are greater than their pre-specified thresholds, respectively, maintenance needs performing. Denote $Y_{(1)}(t), Y_{(3)}(t), ...,Y_{(n)}(t)$ as by sorting the values (realisations) of $Y_1(t), Y_2(t),..., Y_n(t)$ in increasing order. For simplicity, we assume that $Y_k(t)$ are i.i.d for $k=1,2,...,n$ with cdf $F(x,\alpha(t),b^{-1}\beta)$. The cumulative distribution function of $Y_{(r)}(t)$ is given by
 \begin{equation}   
G_{Y_{(r)}(t)}(y)=1-\sum_{k=r}^n\frac{n!}{(n-r)!r!} (1-F(y,\alpha(t),b^{-1}\beta))^k (F(y,\alpha(t),b^{-1}\beta))^{n-k}. \end{equation}   

{\bf First hitting time $T_{L_2}$}. 
Let $T_{L_2}= \textrm{inf}(t > 0: Y_{(r)}(t) \ge {L_2})$. Then the distribution of the first passage time $T_{L_2}$ is given by
\begin{align}    \nonumber
F_{T_{L_2}}(t) &= P(T_{L_2} < t) \\ \nonumber
&= P(Y_{(r)}(t) \ge L_2) \\ \nonumber
&= \sum_{k=r}^n\frac{n!}{(n-r)!r!} (1-F(L_2,\alpha(t),b^{-1}\beta))^k (F(L_2,\alpha(t),b^{-1}\beta))^{n-k},
\end{align}   where $b_k \geq 0$ for all $k$.
\section{Conclusions}
This paper investigated the scenario where a system needs maintenance if a linear combination of the degradation processes exceeds a pre-specified threshold. It derived the probability distribution of the first hitting time and the process of repair cost. The paper then considered the degradation processes that are affected by random effect and covariates. Imperfect repair is conducted when the combined process exceeds a pre-specified threshold, where the imperfect repair is modelled with a geometric process. The system is replaced once the number of its repair reaches a given number. Numerical examples were given to illustrate the maintenance policies derived in the paper.

As our future work, we may investigate the case that a system needs maintenance if $k$ out of $n$ degradation processes exceeds a pre-specified threshold. 


\end{document}